\let\accentvec\vec
\documentclass[runningheads,a4paper,11pt]{llncs}
\let\vec\accentvec

\usepackage[margin=3cm]{geometry}
\usepackage[caption=false,font=footnotesize]{subfig}
\usepackage{amsmath,amssymb,amsfonts,booktabs,xcolor}
\usepackage{graphicx,url,cite,algorithm,algpseudocode}
\usepackage{hyperref,pgfplots}

\algnotext{EndFor}
\algnotext{EndIf}
\algnotext{EndWhile}

\newcommand{\vc}[1]{\boldsymbol{#1}}

\newcommand{\eps}{\varepsilon}
\newcommand{\pr}{\mathbb{P}}

\newcommand{\poly}{\operatorname{poly}}
\newcommand{\cS}{\mathcal{S}}
\newcommand{\cL}{\mathcal{L}}
\newcommand{\Sp}{\mathrm{S}}
\newcommand{\T}{\mathrm{T}}
\newcommand{\rhoq}{\rho_{\mathrm{q}}}

\newcommand{\rhou}{\rho_{\mathrm{u}}}
\newcommand{\alphaq}{\alpha_{\mathrm{q}}}
\newcommand{\alphau}{\alpha_{\mathrm{u}}}

\newcommand{\ip}[2]{\left\langle{#1},{#2}\right\rangle}
\newcommand{\forceindent}{\parindent=2em\indent\parindent=0pt\relax}

\newtheorem{heuristic}{Heuristic}

\usetikzlibrary{calc}

\begin{document}
\mainmatter

\title{Sieving for closest lattice vectors\\ (with preprocessing)}
\titlerunning{Sieving for closest lattice vectors (with preprocessing)}
\toctitle{Lecture Notes in Computer Science}

\author{Thijs Laarhoven}
\authorrunning{Thijs Laarhoven}
\tocauthor{Thijs Laarhoven}

\institute{IBM Research\\R\"{u}schlikon, Switzerland\\\href{mailto:mail@thijs.com}{\UrlFont mail@thijs.com}}

\maketitle

\begin{abstract}
Lattice-based cryptography has recently emerged as a prime candidate for efficient and secure post-quantum cryptography. The two main hard problems underlying its security are the shortest vector problem (SVP) and the closest vector problem (CVP). Various algorithms have been studied for solving these problems, and for SVP, lattice sieving currently dominates in terms of the asymptotic time complexity: one can heuristically solve SVP in time $2^{0.292d + o(d)}$ in high dimensions $d$ [Becker--Ducas--Gama--Laarhoven, SODA'16]. Although several SVP algorithms can also be used to solve CVP, it is not clear whether this also holds for heuristic lattice sieving methods. The best time complexity for CVP is currently $2^{0.377d + o(d)}$ [Becker--Gama--Joux, ANTS'14].

\forceindent In this paper we revisit sieving algorithms for solving SVP, and study how these algorithms can be modified to solve CVP and its variants as well. Our first method is aimed at solving one problem instance and minimizes the overall time complexity for a single CVP instance with a time complexity of $2^{0.292d + o(d)}$. Our second method minimizes the amortized time complexity for several instances on the same lattice, at the cost of a larger preprocessing cost. Using nearest neighbor searching with a balanced space-time tradeoff, with this method we can solve the closest vector problem with preprocessing (CVPP) with $2^{0.636d + o(d)}$ space and preprocessing, in $2^{0.136d + o(d)}$ time, while the query complexity can be further reduced to $2^{0.059d + o(d)}$ at the cost of $2^{d + o(d)}$ space and preprocessing, or even to $2^{\eps d + o(d)}$ for arbitrary $\eps > 0$, at the cost of preprocessing time and memory complexities of $(1/\eps)^{O(d)}$. 

\forceindent For easier variants of CVP, such as approximate CVP and bounded distance decoding (BDD), we further show how the preprocessing method achieves even better complexities. For instance, we can solve approximate CVPP with large approximation factors $\kappa$ with polynomial-sized advice in polynomial time if $\kappa = \Omega(\sqrt{d / \log d})$. This heuristically closes the gap between the decision-CVPP result of [Aharonov--Regev, FOCS'04] (with equivalent $\kappa$) and the search-CVPP result of [Dadush--Regev--Stephens-Davidowitz, CCC'14] (which required larger $\kappa$).


\end{abstract}


\section{Introduction}

\paragraph{\bf Hard lattice problems.} Lattices are discrete subgroups of $\mathbb{R}^d$. More concretely, given a basis $B = \{\vc{b}_1, \dots, \vc{b}_d\} \subset \mathbb{R}^d$, the lattice $\cL = \cL(B)$ generated by $B$ is defined as $\cL(B) = \left\{\sum_{i=1}^d \lambda_i \vc{b}_i: \lambda_i \in \mathbb{Z}\right\}$. Given a basis of a lattice $\cL$, the Shortest Vector Problem (SVP) asks to find a shortest non-zero vector in $\cL$ under the Euclidean norm, i.e., a non-zero lattice vector $\vc{s}$ of norm $\|\vc{s}\| = \lambda_1(\cL) := \min_{\vc{v} \in \cL \setminus \{\vc{0}\}} \|\vc{v}\|$. Given a basis of a lattice and a target vector $\vc{t} \in \mathbb{R}^d$, the Closest Vector Problem (CVP) asks to find a vector $\vc{s} \in \cL$ closest to $\vc{t}$ under the Euclidean distance, i.e.\ such that $\|\vc{s} - \vc{t}\| = \min_{\vc{v} \in \cL} \|\vc{v} - \vc{t}\|$.

These two hard problems are fundamental in the study of lattice-based cryptography, as the security of these schemes is directly related to the hardness of SVP and CVP in high dimensions. Various other hard lattice problems, such as Learning With Errors (LWE) and the Shortest Integer Solution (SIS) problem are closely related to SVP and CVP, and many reductions between these and other hard lattice problems are known; see e.g.\ \cite[Figure 3.1]{laarhoven12kolkata} or \cite{stephens16} for an overview. These reductions show that being able to solve CVP efficiently implies that almost all other lattice problems can also be solved efficiently in the same dimension, which makes the study of the hardness of CVP even more important for choosing parameters in lattice-based cryptography.

\paragraph{\bf Algorithms for SVP and CVP.} Although SVP and CVP are both central in the study of lattice-based cryptography, algorithms for SVP have received somewhat more attention, including a benchmarking website to compare different algorithms~\cite{svp}. Various SVP methods have been studied which can solve CVP as well, such as enumeration (see e.g.\ \cite{kannan83, fincke85, gama10, micciancio15}), discrete Gaussian sampling~\cite{aggarwal15, aggarwal15b}, constructing the Voronoi cell of the lattice~\cite{agrell02, micciancio10}, and using a tower of sublattices~\cite{becker14}. On the other hand, for the asymptotically fastest method in high dimensions for SVP\footnote{To obtain provable guarantees, sieving algorithms are commonly modified to facilitate a somewhat artificial proof technique, which drastically increases the time complexity beyond e.g.\ the discrete Gaussian sampler and the Voronoi cell algorithm~\cite{ajtai01, nguyen08, pujol09, micciancio10b}. On the other hand, if some natural heuristic assumptions are made to enable analyzing the algorithm's behavior, then sieving clearly outperforms these methods. We focus on heuristic sieving in this paper.}, lattice sieving, it is not known how to solve CVP with similar costs as SVP. 

After a series of theoretical works on constructing efficient heuristic sieving algorithms~\cite{nguyen08, micciancio10b, wang11, zhang13, laarhoven15crypto, laarhoven15latincrypt, becker15nns, becker16cp, becker16lsf} as well as practical papers studying how to speed up these algorithms even further~\cite{milde11, schneider11, schneider13, bos14, fitzpatrick14, ishiguro14, mariano14, mariano14b, mariano15, mariano16pdp, mariano16}, the best time complexity for solving SVP currently stands at $2^{0.292d + o(d)}$~\cite{becker16lsf, mariano16}. Although for various other methods the complexities for solving SVP and CVP are similar~\cite{gama10, micciancio10, aggarwal15b}, one can only guess whether the same holds for lattice sieving methods. To date, the best heuristic time complexity for solving CVP in high dimensions stands at $2^{0.377d + o(d)}$, due to Becker--Gama--Joux~\cite{becker14}.

\subsection{Contributions}

In this paper we revisit heuristic lattice sieving algorithms, as well as the recent trend to speed up these algorithms using nearest neighbor searching, and we investigate how these algorithms can be modified to solve CVP and its generalizations. We present two different approaches for solving CVP with sieving, each of which we argue has its own merits.

\paragraph{\bf Adaptive sieving.} In \textit{adaptive sieving}, we adapt the entire sieving algorithm to the problem instance, including the target vector. As the resulting algorithm is tailored specifically to the given CVP instance, this leads to the best asymptotic complexity for solving a single CVP instance out of our two proposed methods: $2^{0.292d + o(d)}$ time and space. This method is very similar to solving SVP with lattice sieving, and leads to equivalent asymptotics on the space and time complexities as for SVP. The corresponding space-time tradeoff is illustrated in Figure~\ref{fig:1}, and equals that of \cite{becker16lsf} for solving SVP.

\paragraph{\bf Non-adaptive sieving.} Our main contribution, \textit{non-adaptive sieving}, takes a different approach, focusing on cases where several CVP instances are to be solved on the same lattice. The goal here is to minimize the costs of computations depending on the target vector, and spend more time on preprocessing the lattice, so that the amortized time complexity per instance is smaller when solving many CVP instances on the same lattice. This is very closely related to the Closest Vector Problem with Preprocessing (CVPP), where the difference is that we allow for exponential-size preprocessed space. Using nearest neighbor techniques with a balanced space-time tradeoff, we show how to solve CVPP with $2^{0.636d + o(d)}$ space and preprocessing, in $2^{0.136d + o(d)}$ time. A continuous tradeoff between the two complexities can be obtained, where in the limit we can solve CVPP with $(1/\eps)^{O(d)}$ space and preprocessing, in $2^{\eps d + o(d)}$ time. This tradeoff is depicted in Figure~\ref{fig:1}.

A potential application of non-adaptive sieving is as a subroutine within enumeration methods. As described in e.g.\ \cite{gama10}, at any given level in the enumeration tree, one is attempting to solve a CVP instance in a lower-dimensional sublattice of $\cL$, where the target vector is determined by the path chosen from the root to the current node in the tree. That means that if we can preprocess this sublattice such that the amortized time complexity of solving CVPP is small, then this could speed up processing the bottom part of the enumeration tree. This in turn might help speed up the lattice basis reduction algorithm BKZ~\cite{schnorr87, schnorr94, chen11}, which commonly uses enumeration as its SVP subroutine, and is key in assessing the security of lattice-based schemes. As the preprocessing needs to be performed once, CVPP algorithms with impractically large preprocessing costs may not be useful, but we show that with sieving the preprocessing costs can be quite small.

\begin{figure}[!t]
{\center
\includegraphics{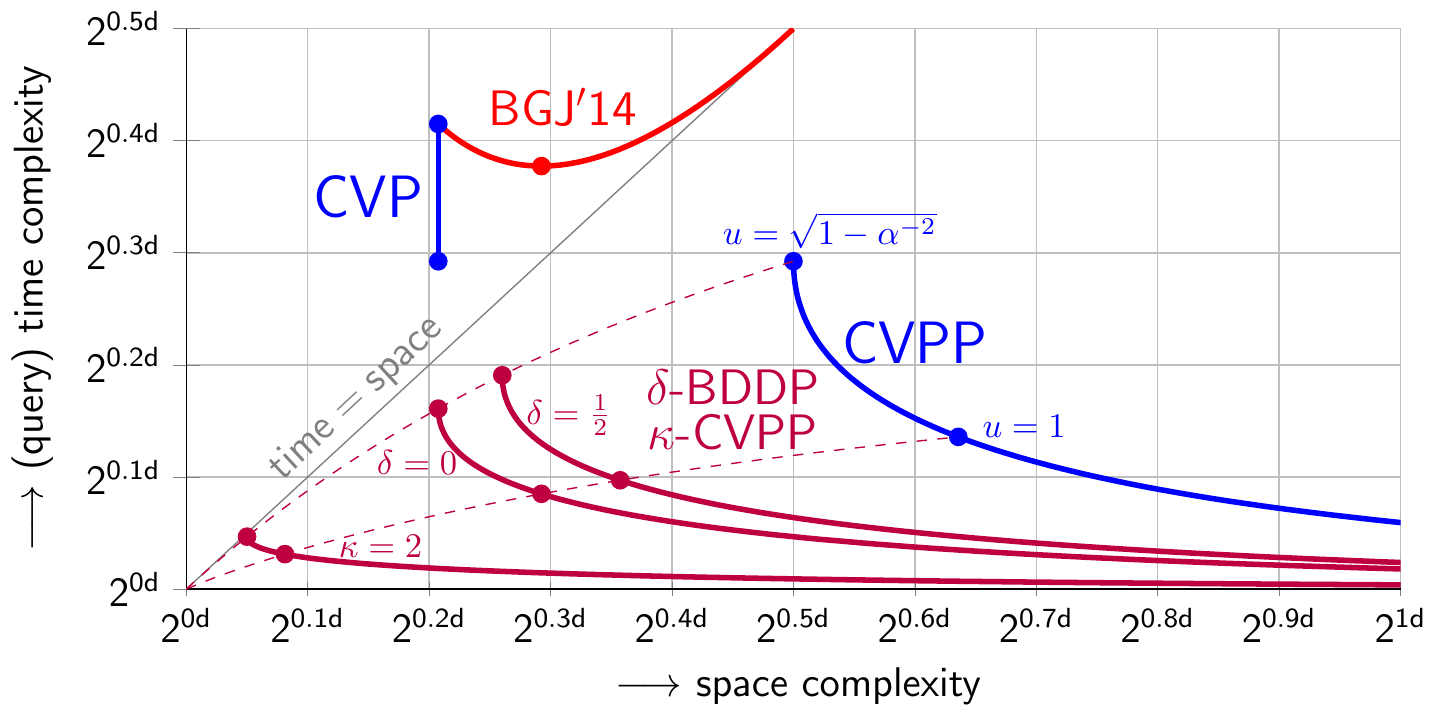}}
\caption{Heuristic complexities for solving the Closest Vector Problem (CVP), the Closest Vector Problem with Preprocessing (CVPP), Bounded Distance Decoding with Preprocessing ($\delta$-BDDP), and the Approximate Closest Vector Problem with Preprocessing ($\kappa$-CVPP). The red curve shows CVP complexities of Becker--Gama--Joux~\cite{becker14}. The left blue curve denotes CVP complexities of adaptive sieving. The right blue curve shows exact CVPP complexities using non-adaptive sieving. Purple curves denote relaxations of CVPP corresponding to different parameters $\delta$ (BDD radius) and $\kappa$ (approximation factor). Note that exact CVPP corresponds to $\delta$-BDDP with $\delta = 1$ and to $\kappa$-CVPP with $\kappa = 1$.\label{fig:1}}
\end{figure}

\paragraph{\bf Outline.} The remainder of the paper is organized as follows. In Section~\ref{sec:pre} we describe some preliminaries, such as sieving algorithms and a useful result on nearest neighbor searching. Section~\ref{sec:ad} describes adaptive sieving and its analysis for solving CVP without preprocessing. Section~\ref{sec:non} describes the preprocessing approach to solving CVP, with complexity analyses for exact CVP and some of its relaxations. 


\section{Preliminaries}
\label{sec:pre}

\subsection{Lattice sieving for solving SVP}

Heuristic lattice sieving algorithms for solving the shortest vector problem all use the following basic property of lattices: if $\vc{v}, \vc{w} \in \cL$, then their sum/difference $\vc{v} \pm \vc{w} \in \cL$ is a lattice vector as well. Therefore, if we have a long list $L$ of lattice vectors stored in memory, we can consider combinations of these vectors to obtain new, shorter lattice vectors. To make sure the algorithm makes progress in finding shorter lattice vectors, $L$ needs to contain a lot of lattice vectors; for vectors $\vc{v}, \vc{w} \in \cL$ of similar norm, the vector $\vc{v} - \vc{w}$ is shorter than $\vc{v}, \vc{w}$ iff the angle between $\vc{v}, \vc{w}$ is smaller than $\pi/3$, which for random vectors $\vc{v}, \vc{w}$ occurs with probability $(3/4)^{d/2 + o(d)}$. The expected space complexity of heuristic sieving algorithms follows directly from this observation: if we draw $(4/3)^{d/2 + o(d)}$ random vectors from the unit sphere, we expect a large number of pairs of vectors to have angle less than $\pi/3$, leading to many short difference vectors. This is exactly the heuristic assumption used in analyzing these sieving algorithms: when normalized, vectors in $L$ follow the same distribution as vectors sampled uniformly at random from the unit sphere.
\begin{heuristic} \label{heur:1}
When normalized, the list vectors $\vc{w} \in L$ behave as i.i.d.\ uniformly distributed random vectors from the unit sphere $\cS^{d-1} := \{\vc{x} \in \mathbb{R}^d: \|\vc{x}\| = 1\}$. 
\end{heuristic}
Therefore, if we start by sampling a list $L$ of $(4/3)^{d/2 + o(d)}$ long lattice vectors, and iteratively consider combinations of vectors in $L$ to find shorter vectors, we expect to keep making progress. Note that naively, combining pairs of vectors in a list of size $(4/3)^{d/2 + o(d)} \approx 2^{0.208d + o(d)}$ takes time $(4/3)^{d + o(d)} \approx 2^{0.415d + o(d)}$.

\paragraph{\bf The Nguyen-Vidick sieve.} The heuristic sieve algorithm of Nguyen and Vidick~\cite{nguyen08} starts by sampling a list $L$ of $(4/3)^{d/2 + o(d)}$ long lattice vectors, and uses a \textit{sieve} to map $L$, with maximum norm $R := \max_{\vc{v} \in L} \|\vc{v}\|$, to a new list $L'$, with maximum norm at most $\gamma R$ for $\gamma < 1$ close to $1$. By repeatedly applying this sieve, after $\poly(d)$ iterations we expect to find a long list of lattice vectors of norm at most $\gamma^{\poly(d)} R = O(\lambda_1(\cL))$. The final list is then expected to contain a shortest vector of the lattice. Algorithm~\ref{alg:nv} in Appendix~\ref{app:alg} describes a sieve equivalent to Nguyen-Vidick's original sieve, to map $L$ to $L'$ in $|L|^2$ time.

\paragraph{\bf Micciancio and Voulgaris' GaussSieve.} Micciancio and Voulgaris used a slightly different approach in the GaussSieve~\cite{micciancio10b}. This algorithm reduces the memory usage by immediately \textit{reducing} all pairs of lattice vectors that are sampled. The algorithm uses a single list $L$, which is always kept in a state where for all $\vc{w}_1, \vc{w}_2 \in L$, $\|\vc{w}_1 \pm \vc{w}_2\| \geq \|\vc{w}_1\|, \|\vc{w}_2\|$, and each time a new vector $\vc{v} \in \cL$ is sampled, its norm is reduced with vectors in $L$. After the norm can no longer be reduced, the vectors in $L$ are reduced with $\vc{v}$. Modified list vectors are added to a stack to be processed later (to maintain the pairwise reduction-property of $L$), and new vectors which are pairwise reduced with $L$ are added to $L$. Immediately reducing all pairs of vectors means that the algorithm uses less time and memory in practice, but at the same time Nguyen and Vidick's heuristic proof technique does not apply here. However, it is commonly believed that the same bounds $(4/3)^{d/2 + o(d)}$ and $(4/3)^{d + o(d)}$ on the space and time complexities hold for the GaussSieve. Pseudocode of the GaussSieve is given in Algorithm~\ref{alg:gauss} in Appendix~\ref{app:alg}.

\subsection{Nearest neighbor searching}

Given a data set $L \subset \mathbb{R}^d$, the nearest neighbor problem asks to preprocess $L$ such that, when given a query $\vc{t} \in \mathbb{R}^d$, one can quickly return a nearest neighbor $\vc{s} \in L$ with distance $\|\vc{s} - \vc{t}\| = \min_{\vc{w} \in L} \|\vc{w} - \vc{t}\|$. This problem is essentially identical to CVP, except that $L$ is a finite set of unstructured points, rather than the infinite set of all points in a lattice $\cL$.

\paragraph{\bf Locality-Sensitive Hashing/Filtering (LSH/LSF).} A celebrated technique for finding nearest neighbors in high dimensions is Locality-Sensitive Hashing (LSH)~\cite{indyk98, wang14}, where the idea is to construct many random partitions of the space, and store the list $L$ in hash tables with buckets corresponding to regions. Preprocessing then consists of constructing these hash tables, while a query $\vc{t}$ is answered by doing a lookup in each of the hash tables, and searching for a nearest neighbor in these buckets. More details on LSH in combination with sieving can be found in e.g.\ \cite{laarhoven15crypto, laarhoven15latincrypt, becker15nns, becker16cp}.

Similar to LSH, Locality-Sensitive Filtering (LSF)~\cite{becker16lsf, laarhoven15nns} divides the space into regions, with the added relaxation that these regions do not have to form a partition; regions may overlap, and part of the space may not be covered by any region. This leads to improved results compared to LSH when $L$ has size exponential in $d$~\cite{becker16lsf, laarhoven15nns}. Below we restate one of the main results of~\cite{laarhoven15nns} for our applications. The specific problem considered here is: given a data set $L \subset \cS^{d-1}$ sampled uniformly at random, and a random query $\vc{t} \in \cS^{d-1}$, return a vector $\vc{w} \in L$ such that the angle between $\vc{w}$ and $\vc{t}$ is at most $\theta$. The following result further assumes that the list $L$ contains $n = (1 / \sin \theta)^{d + o(d)}$ vectors.

\begin{lemma} \label{lem:nns} \cite[Corollary 1]{laarhoven15nns}
Let $\theta \in (0, \frac{1}{2} \pi)$, and let $u \in [\cos \theta, 1/\cos \theta]$. Let $L \subset \cS^{d-1}$ be a list of $n = (1 / \sin \theta)^{d + o(d)}$ vectors sampled uniformly at random from $\cS^{d-1}$. Then, using spherical LSF with parameters $\alphaq = u \cos \theta$ and $\alphau = \cos \theta$, one can preprocess $L$ in time $n^{1 + \rhou + o(1)}$, using $n^{1 + \rhou + o(1)}$ space, and with high probability answer a random query $\vc{t} \in \cS^{d-1}$ correctly in time $n^{\rhoq + o(1)}$, where: 
	\begin{align}
	n^{\rhoq} &= \left(\frac{\sin^2 \theta \, (u \cos \theta + 1)}{u  \cos \theta - \cos 2 \theta}\right)^{d/2}, \quad n^{\rhou} = \left(\frac{\sin^2 \theta}{1 - \cot^2 \theta\left(u^2 - 2 u \cos \theta + 1\right)}\right)^{d/2}. \label{eq:main3}
	\end{align}
\end{lemma}

Applying this result to sieving for solving SVP, where $n = \sin(\frac{\pi}{3})^{-d + o(d)}$ and we are looking for pairs of vectors at angle at most $\frac{\pi}{3}$ to perform reductions, this leads to a space and preprocessing complexity of $n^{0.292d + o(d)}$, and a query complexity of $2^{0.084d + o(d)}$. As the preprocessing in sieving is only performed once, and queries are performed $n \approx 2^{0.208d + o(d)}$ times, this leads to a reduction of the complexities of sieving (for SVP) from $2^{0.208d + o(d)}$ space and $2^{0.415d + o(d)}$ time, to $2^{0.292d + o(d)}$ space and time~\cite{becker16lsf}.


\section{Adaptive sieving for CVP}
\label{sec:ad}

We present two methods for solving CVP using sieving, the first of which we call \textit{adaptive sieving} -- we adapt the entire sieving algorithm to the particular CVP instance, to obtain the best overall time complexity for solving one instance. When solving several CVP instances, the costs roughly scale linearly with the number of instances.

\subsubsection{Using one list.} The main idea behind this method is to translate the SVP algorithm by the target vector $\vc{t}$; instead of generating a long list of lattice vectors reasonably close to $\vc{0}$, we generate a list of lattice vectors close to $\vc{t}$, and combine lattice vectors to find lattice vectors even closer vectors to $\vc{t}$. The final list then hopefully contains a closest vector to $\vc{t}$. 

One quickly sees that this does not work, as the fundamental property of lattices does not hold for the lattice coset $\vc{t} + \cL$: if $\vc{w}_1, \vc{w}_2 \in \vc{t} + \cL$, then $\vc{w}_1 \pm \vc{w}_2 \notin \vc{t} + \cL$. In other words, two lattice vectors close to $\vc{t}$ can only be combined to form lattice vectors close to $\vc{0}$ or $2 \vc{t}$. So if we start with a list of vectors close to $\vc{t}$, and combine vectors in this list as in the Nguyen-Vidick sieve, then after one iteration we will end up with a list $L'$ of lattice vectors close to $\vc{0}$.

\subsubsection{Using two lists.} To make the idea of translating the whole problem by $\vc{t}$ work for the Nguyen-Vidick sieve, we make the following modification: we keep track of two lists $L = L_{\vc{0}}$ and $L_{\vc{t}}$ of lattice vectors close to $\vc{0}$ and $\vc{t}$, and construct a sieve which maps two input lists $L_{\vc{0}}, L_{\vc{t}}$ to two output lists $L_{\vc{0}}', L_{\vc{t}}'$ of lattice vectors slightly closer to $\vc{0}$ and $\vc{t}$. Similar to the original Nguyen-Vidick sieve, we then apply this sieve several times to two initial lists $(L_{\vc{0}}, L_{\vc{t}})$ with a large radius $R$, to end up with two lists $L_{\vc{0}}$ and $L_{\vc{t}}$ of lattice vectors at distance at most approximately $\sqrt{4/3} \cdot \lambda_1(\cL)$ from $\vc{0}$ and $\vc{t}$\footnote{Observe that by the Gaussian heuristic, there are $(4/3)^{d/2 + o(d)}$ vectors in $\cL$ within any ball of radius $\sqrt{4/3} \cdot \lambda_1(\cL)$. So the list size of the NV-sieve will surely decrease below $(4/3)^{d/2}$ when $R < \sqrt{4/3} \cdot \lambda_1(\cL)$.}. The argumentation that this algorithm works is almost identical to that for solving SVP, where we now make the following slightly different heuristic assumption.
\begin{heuristic} \label{heur:2}
When normalized, the list vectors $L_{\vc{0}}$ and $L_{\vc{t}}$ in the modified Nguyen-Vidick sieve both behave as i.i.d.\ uniformly distributed random vectors from the unit sphere. 
\end{heuristic}
The resulting algorithm, based on the Nguyen-Vidick sieve, is presented in Algorithm~\ref{alg:nv-adaptive}. 

\begin{algorithm}[!t]
\caption{The adaptive Nguyen-Vidick sieve for finding closest vectors}
\label{alg:nv-adaptive}
\begin{algorithmic}[1]
\Require Lists $L_{\vc{0}}, L_{\vc{t}} \subset \cL$ containing $(4/3)^{d/2 + o(d)}$ vectors at distance $\leq R$ from $\vc{0}, \vc{t}$
\Ensure Lists $L_{\vc{0}}', L_{\vc{t}}' \subset \cL$ contain $(4/3)^{d/2 + o(d)}$ vectors at distance $\leq \gamma R$ from $\vc{0}, \vc{t}$
\State Initialize empty lists $L_{\vc{0}}', L_{\vc{t}}'$ 
\For{\textbf{each} $(\vc{w}_1, \vc{w}_2) \in L_{\vc{0}} \times L_{\vc{0}}$}
	\If{$\|\vc{w}_1 - \vc{w}_2\| \leq \gamma R$}
		\State Add $\vc{w}_1 - \vc{w}_2$ to the list $L_{\vc{0}}'$
	\EndIf
\EndFor
\For{\textbf{each} $(\vc{w}_1, \vc{w}_2) \in L_{\vc{t}} \times L_{\vc{0}}$}
	\If{$\|(\vc{w}_1 - \vc{w}_2) - \vc{t}\| \leq \gamma R$}
		\State Add $\vc{w}_1 - \vc{w}_2$ to the list $L_{\vc{t}}'$
	\EndIf
\EndFor
\State \Return $(L_{\vc{0}}', L_{\vc{t}}')$ 
\end{algorithmic}
\end{algorithm}

\subsubsection{Main result.} As the (heuristic) correctness of this algorithm follows directly from the correctness of the original NV-sieve, and nearest neighbor techniques can be applied to this algorithm in similar fashion as well, we immediately obtain the following result. Note that space-time tradeoffs for SVP, such as the one illustrated in \cite[Figure 1]{becker16lsf}, similarly carry over to solving CVP, and the best tradeoff for SVP (and therefore CVP) is depicted in Figure~\ref{fig:1}. 

\begin{theorem}
Assuming Heuristic~\ref{heur:2} holds, the adaptive Nguyen-Vidick sieve with spherical LSF solves CVP in time $\T$ and space $\Sp$, with
\begin{align}
\Sp = (4/3)^{d/2 + o(d)} \approx 2^{0.208 d + o(d)}, \quad \T = (3/2)^{d/2 + o(d)} \approx 2^{0.292 d + o(d)}.
\end{align}
\end{theorem}

An important open question is whether these techniques can also be applied to the faster GaussSieve algorithm to solve CVP. The GaussSieve seems to make even more use of the property that the sum/difference of two lattice vectors is also in the lattice, and operations in the GaussSieve in $\cL$ cannot as easily be \textit{mimicked} for the coset $\vc{t} + \cL$. Solving CVP with the GaussSieve with similar complexities is left as an open problem.


\section{Non-adaptive sieving for CVPP}
\label{sec:non}

Our second method for finding closest vectors with heuristic lattice sieving follows a slightly different approach. Instead of focusing only on the total time complexity for one problem instance, we split the algorithm into two phases:
\begin{itemize}
	\item Phase 1: Preprocess the lattice $\cL$, without knowledge of the target $\vc{t}$;
	\item Phase 2: Process the query $\vc{t}$ and output a closest lattice vector $\vc{s} \in \cL$ to $\vc{t}$.
\end{itemize}

Intuitively it may be more important to keep the costs of Phase 2 small, as the preprocessed data can potentially be reused later for other instances on the same lattice. This approach is essentially equivalent to the Closest Vector Problem with Preprocessing (CVPP): preprocess $\cL$ such that when given a target vector $\vc{t}$ later, one can quickly return a closest vector $\vc{s} \in \cL$ to $\vc{t}$. For CVPP however the preprocessed space is usually restricted to be of polynomial size, and the time used for preprocessing the lattice is often not taken into account. Here we will keep track of the preprocessing costs as well, and we do not restrict the output from the preprocessing phase to be of size $\poly(d)$.

\subsubsection{Algorithm description.} To minimize the costs of answering a query, and to do the preprocessing independently of the target vector, we first run a standard SVP sieve, resulting in a large list $L$ of almost all short lattice vectors. Then, after we are given the target vector $\vc{t}$, we use $L$ to reduce the target. Finally, once the resulting vector $\vc{t}' \in \vc{t} + \cL$ can no longer be reduced with our list, we hope that this reduced vector $\vc{t}'$ is the shortest vector in the coset $\vc{t} + \cL$, so that $\vc{0}$ is the closest lattice vector to $\vc{t}'$ and $\vc{s} = \vc{t} - \vc{t}'$ is the closest lattice vector to $\vc{t}$. 

The first phase of this algorithm consists in running a sieve and storing the resulting list in memory (potentially in a nearest neighbor data structure for faster lookups). For this phase either the Nguyen-Vidick sieve or the GaussSieve can be used. The second phase is the same for either method, and is described in Algorithm~\ref{alg:nonadaptive} for the general case of an input list essentially consisting of the $\alpha^{d + o(d)}$ shortest vectors in the lattice. Note that a standard SVP sieve would produce a list of size $(4/3)^{d/2 + o(d)}$ corresponding to $\alpha = \sqrt{4/3}$. 

\begin{algorithm}[!t]
\caption{Non-adaptive sieving (Phase 2) for finding closest vectors}
\label{alg:nonadaptive}
\begin{algorithmic}[1]
\Require A list $L \subset \cL$ of $\alpha^{d/2 + o(d)}$ vectors of norm at most $\alpha \cdot \lambda_1(\cL)$, and $\vc{t} \in \mathbb{R}^d$
\Ensure The output vector $\vc{s}$ is the closest lattice vector to $\vc{t}$ (w.h.p.)
\State Initialize $\vc{t}' \leftarrow \vc{t}$
\For{\textbf{each} $\vc{w} \in L$}
	\If{$\|\vc{t}' - \vc{w}\| \leq \|\vc{t}'\|$}
		\State Replace $\vc{t}' \leftarrow \vc{t}' - \vc{w}$ and restart the \textbf{for}-loop
	\EndIf
\EndFor
\State \Return $\vc{s} = \vc{t} - \vc{t}'$
\end{algorithmic}
\end{algorithm}

\subsubsection{List size.} We first study how large $L$ must be to guarantee that the algorithm succeeds. One might wonder why we do not fix $\alpha = \sqrt{4/3}$ immediately in Algorithm~\ref{alg:nonadaptive}. To see why this choice of $\alpha$ does not suffice, suppose we have a vector $\vc{t}' \in \vc{t} + \cL$ which is no longer reducible with $L$. This implies that $\vc{t}'$ has norm approximately $\sqrt{4/3} \cdot \lambda_1(\cL)$, similar to what happens in the GaussSieve. Now, unfortunately the fact that $\vc{t}'$ cannot be reduced with $L$ anymore, does \textit{not} imply that the closest lattice point to $\vc{t}'$ is $\vc{0}$. In fact, it is more likely that there exists an $\vc{s} \in \vc{t} + \cL$ of norm slightly more than $\sqrt{4/3} \cdot \lambda_1(\cL)$ which is closer to $\vc{t}'$, but which is not used for reductions. 

By the Gaussian heuristic, we expect the distance from $\vc{t}$ and $\vc{t}'$ to the lattice to be $\lambda_1(\cL)$. So to guarantee that $\vc{0}$ is the closest lattice vector to the reduced vector $\vc{t}'$, we need $\vc{t}'$ to have norm at most $\lambda_1(\cL)$. To analyze and prove correctness of this algorithm, we will therefore prove that, under the assumption that the input is a list of the $\alpha^{d + o(d)}$ shortest lattice vectors of norm at most $\alpha \cdot \lambda_1(\cL)$ for a particular choice of $\alpha$, w.h.p.\ the algorithm reduces $\vc{t}$ to a vector $\vc{t}' \in \vc{t} + \cL$ of norm at most $\lambda_1(\cL)$.

To study how to set $\alpha$, we start with the following elementary lemma regarding the probability of reduction between two uniformly random vectors with given norms.

\begin{lemma} \label{lem:1}
Let $v, w > 0$ and let $\vc{v} = v \cdot \vc{e}_v$ and $\vc{w} = w \cdot \vc{e}_w$. Then:
\begin{align}
\pr_{\vc{e}_v, \vc{e}_w \sim \cS^{d-1}}\Big(\|\vc{v} - \vc{w}\|^2 < \|\vc{v}\|^2\Big) \sim \left[1 - \left(\tfrac{w}{2v}\right)^2\right]^{d/2 + o(d)}.
\end{align}
\end{lemma}

\begin{proof}
Expanding $\|\vc{v} - \vc{w}\|^2 = v^2 + w^2 - 2 v w \ip{\vc{e}_v}{\vc{e}_w}$ and $\|\vc{v}\|^2 = v^2$, the condition $\|\vc{v} - \vc{w}\|^2 < \|\vc{v}\|^2$ equals $\frac{w}{2v} < \ip{\vc{e}_v}{\vc{e}_w}$. The result follows from \cite[Lemma 2.1]{becker16lsf}.
\end{proof}

Under Heuristic~\ref{heur:1}, we then obtain a relation between the choice of $\alpha$ for the input list and the expected norm of the reduced vector $\vc{t}'$ as follows.

\begin{lemma} \label{lem:2}
Let $L \subset \alpha \cdot \cS^{d-1}$ be a list of $\alpha^{d + o(d)}$ uniformly random vectors of norm $\alpha > 1$, and let $\vc{v} \in \beta \cdot \cS^{d-1}$ be sampled uniformly at random. Then, for high dimensions $d$, there exists a $\vc{w} \in L$ such that $\|\vc{v} - \vc{w}\| < \|\vc{v}\|$ if and only if
\begin{align}
\alpha^4 - 4 \beta^2 \alpha^2 + 4\beta^2 
< 0. \label{eq:a}
\end{align}
\end{lemma}

\begin{proof}
By Lemma~\ref{lem:1} we can reduce $\vc{v}$ with $\vc{w} \in L$ with probability similar to $p = [1 - \frac{\alpha^2}{4\beta^2}]^{d/2 + o(d)}$. Since we have $n = \alpha^{d + o(d)}$ such vectors $\vc{w}$, the probability that none of them can reduce $\vc{v}$ is $(1 - p)^n$, which is $o(1)$ if $n \gg 1/p$ and $1 - o(1)$ if $n \ll 1/p$. Expanding $n \cdot p$, we obtain the given equation~\eqref{eq:a}, where $\alpha^4 - 4 \beta^2 \alpha^2 + 4 \beta^2 < 0$ implies $n \gg 1/p$.
\end{proof}

Note that in our applications, we do not just have a list of $\alpha^{d + o(d)}$ lattice vectors of norm $\alpha \cdot \lambda_1(\cL)$; for any $\alpha_0 \in [1, \alpha]$ we expect $L$ to contain $\alpha_0^{d + o(d)}$ lattice vectors of norm at most $\alpha_0 \cdot \lambda_1(\cL)$. To obtain a reduced vector $\vc{t}'$ of norm $\beta \cdot \lambda_1(\cL)$, we therefore obtain the condition that for \textit{some} value $\alpha_0 \in [1, \alpha]$, it must hold that $\alpha_0^4 - 4 \beta^2 \alpha_0^2 + 4\beta_0^2 < 0$. 

From~\eqref{eq:a} it follows that $p(\alpha^2) = \alpha^4 - 4 \beta^2 \alpha^2 + 4\beta^2$ has two roots $r_1 < 2 < r_2$ for $\alpha^2$, which lie close to $2$ for $\beta \approx 1$. The condition that $p(\alpha_0^2) < 0$ for some $\alpha_0 \leq \alpha$ is equivalent to $\alpha > r_1$, which for $\beta = 1 + o(1)$ implies that $\alpha^2 \geq 2 + o(1)$. This means that asymptotically we must set $\alpha = \sqrt{2}$, and use $n = 2^{d/2 + o(d)}$ input vectors, to guarantee that w.h.p.\ the algorithm succeeds. A sketch of the situation is also given in Figure~\ref{fig:2a}.

\begin{figure}[!t]
\subfloat[For solving \textbf{exact CVP}, we must reduce the vector $\vc{t}$ to a vector $\vc{t}' \in \vc{t} + \cL$ of norm at most $\lambda_1(\cL)$. The nearest lattice point to $\vc{t}'$ lies in a ball of radius approximately $\lambda_1(\cL)$ around $\vc{t}'$ (blue), and almost all the mass of this ball is contained in the (black) ball around $\vc{0}$ of radius $\sqrt{2} \cdot \lambda_1(\cL)$. So if $\vc{s} \in \cL \setminus \{\vc{0}\}$ had lain closer to $\vc{t}'$ than $\vc{0}$, we would have reduced $\vc{t}'$ with $\vc{s}$, since $\vc{s} \in L$.\label{fig:2a}]{%
\includegraphics{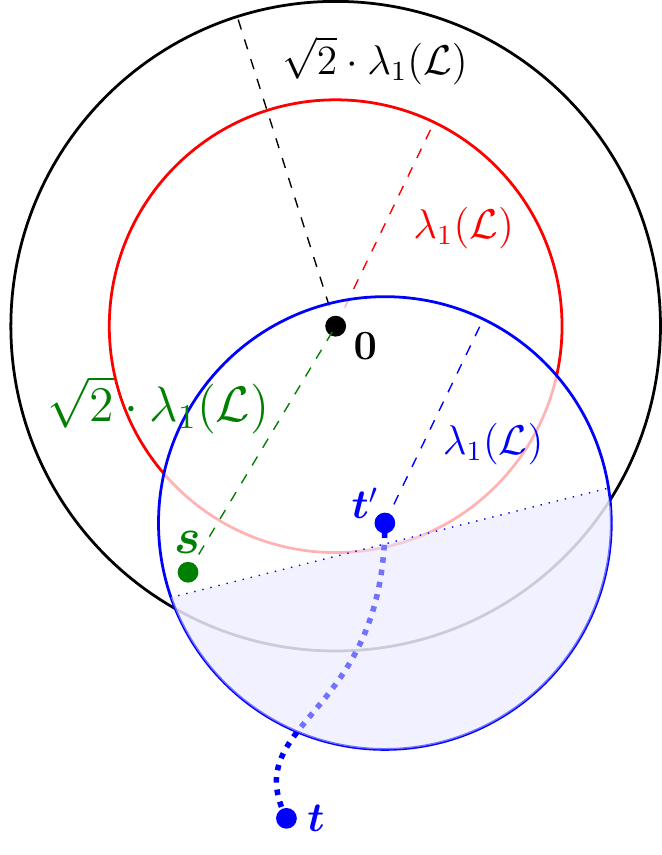}}%
\hfill
\subfloat[For \textbf{variants of CVP}, a choice $\alpha$ for the list size implies a norm $\beta \cdot \lambda_1(\cL)$ of $\vc{t}'$. The nearest lattice vector $\vc{s}$ to $\vc{t}'$ lies within $\delta \cdot \lambda_1(\cL)$ of $\vc{t}'$ ($\delta = 1$ for approx-CVP), so with high probability $\vc{s}$ has norm approximately $(\sqrt{\beta^2 + \delta^2}) \cdot \lambda_1(\cL)$. For $\delta$-BDD, if $\sqrt{\beta^2 + \delta^2} \leq \alpha$ then we expect the nearest point $\vc{s}$ to be in the list $L$. For $\kappa$-CVP, if $\beta \leq \kappa$, then the lattice vector $\vc{t} - \vc{t}'$ has norm at most $\kappa \cdot \lambda_1(\cL)$.\label{fig:2b}]{
\includegraphics{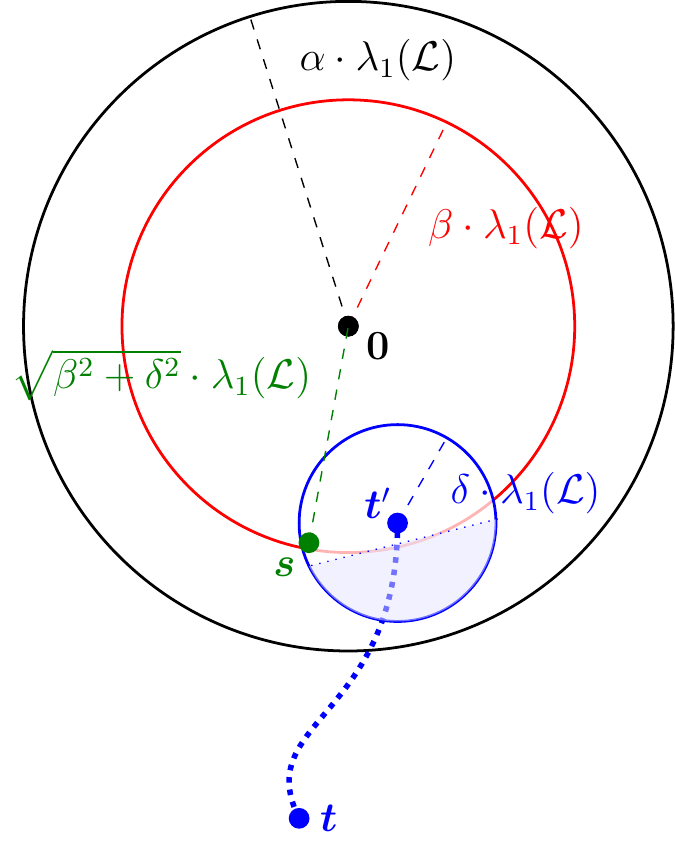}}%
\caption{Comparison of the list size complexity analysis for CVP (left) and BDD/approximate CVP (right). The point $\vc{t}$ represents the target vector, and after a series of reductions with Algorithm~\ref{alg:nonadaptive}, we obtain $\vc{t}' \in \vc{t} + \cL$. Blue balls around $\vc{t}'$ depict regions in which we expect the closest lattice point to $\vc{t}'$ to lie, where the blue shaded area indicates a negligible fraction of this ball~\cite[Lemma 2]{becker16lsf}.\label{fig:2}}
\end{figure}

\subsubsection{Modifying the first phase.} As we will need a larger list of size $2^{d/2 + o(d)}$ to make sure we can solve CVP exactly, we need to adjust Phase 1 of the algorithm as well. Recall that with standard sieving, we reduce vectors iff their angle is at most $\theta = \frac{\pi}{3}$, resulting in a list of size $(\sin \theta)^{-d + o(d)}$. As we now need the output list of the first phase to consist of $2^{d/2 + o(d)} = (\sin \theta')^{-d + o(d)}$ vectors for $\theta' = \frac{\pi}{4}$, we make the following adjustment: only reduce $\vc{v}$ and $\vc{w}$ if their common angle is less than $\frac{\pi}{4}$. For unit length vectors, this condition is equivalent to reducing $\vc{v}$ with $\vc{w}$ iff $\|\vc{v} - \vc{w}\|^2 \leq (2 - \sqrt{2}) \cdot \|\vc{v}\|^2$. This further accelerates nearest neighbor techniques due to the smaller angle $\theta$. Pseudocode for the modified first phase is given in Appendix~\ref{app:alg2}

\subsubsection{Main result.} With the algorithm in place, let us now analyze its complexity for solving CVP. The first phase of the algorithm generates a list of size $2^{d/2 + o(d)}$ by combining pairs of vectors, and naively this can be done in time $\T_1 = 2^{d + o(d)}$ and space $\Sp = 2^{d/2 + o(d)}$, with query complexity $\T_2 = 2^{d/2 + o(d)}$. Using nearest neighbor searching (Lemma~\ref{lem:nns}), the query and preprocessing complexities can be further reduced, leading to the following result.

\begin{theorem} \label{thm:2}
Let $u \in (\frac{1}{2} \sqrt{2}, \sqrt{2})$. Using non-adaptive sieving, we can solve CVP with preprocessing time $\T_1$, space complexity $\Sp$, and query time complexity $\T_2$ as follows:
\begin{align}
\Sp = \T_1 &= \left(\frac{1}{u (\sqrt{2} - u)}\right)^{d/2 + o(d)}, \qquad \T_2 = \left(\frac{\sqrt{2} + u}{2 u}\right)^{d/2 + o(d)}.
\end{align}
\end{theorem}

\begin{proof}
These complexities follow from Lemma~\ref{lem:nns} with $\theta = \frac{\pi}{4}$, noting that the first phase can be performed in time and space $\T_1 = \Sp = n^{1 + \rhou}$, and the second phase in time $\T_2 = n^{\rhoq}$.
\end{proof}

To illustrate the time and space complexities of Theorem~\ref{thm:2}, we highlight three special cases $u$ as follows. The full tradeoff curve for $u \in (\frac{1}{2} \sqrt{2}, \sqrt{2})$ is depicted in Figure~\ref{fig:1}.
\begin{itemize}
	\item Setting $u = \frac{1}{2} \sqrt{2}$, we obtain $\Sp = \T_1 = 2^{d/2 + o(d)}$ and $\T_2 
	\approx 2^{0.2925d + o(d)}$.	
	\item Setting $u = 1$, we obtain $\Sp = \T_1 
	\approx 2^{0.6358 d + o(d)}$ and $\T_2 
	\approx 2^{0.1358 d + o(d)}$. 
	\item Setting $u = \frac{1}{2}(\sqrt{2} + 1)$, we get $\Sp = \T_1 = 2^{d + o(d)}$ and $\T_2 
	\approx 2^{0.0594 d + o(d)}$.
\end{itemize}
The first result shows that the query complexity of non-adaptive sieving is never worse than for adaptive sieving; only the space and preprocessing complexities are worse. The second and third results show that CVP can be solved in significantly less time, even with preprocessing and space complexities bounded by $2^{d + o(d)}$. 

\paragraph{\bf Minimizing the query complexity.} As $u \to \sqrt{2}$, the query complexity keeps decreasing while the memory and preprocessing costs increase. For arbitrary $\eps > 0$, we can set $u = u_\eps \approx \sqrt{2}$ as a function of $\eps$, resulting in asymptotic complexities $\Sp = \T_1 = (1/\eps)^{O(d)}$ and $\T_2 = 2^{\eps d + o(d)}$. This shows that it is possible to obtain a slightly subexponential query complexity, at the cost of superexponential space, by taking $\eps = o(1)$ as a function of $d$. 

\begin{corollary} \label{thm:3}
For arbitrary $\eps > 0$, using non-adaptive sieving we can solve CVPP with preprocessing time and space complexities $(1/\eps)^{O(d)}$, in time $2^{\eps d + o(d)}$. In particular, we can solve CVPP in $2^{o(d)}$ time, using $2^{\omega(d)}$ space and preprocessing.
\end{corollary}

Being able to solve CVPP in subexponential time with superexponential preprocessing and memory is neither trivial nor quite surprising. A naive approach to the problem, with this much memory, could for instance be to index the entire fundamental domain of $\cL$ in a hash table. One could partition this domain into small regions, solve CVP for the centers of each of these regions, and store all the solutions in memory. Then, given a query, one looks up which region $\vc{t}$ is in, and returns the answer corresponding to that vector. With a sufficiently fine-grained partitioning of the fundamental domain, the answers given by the look-ups are accurate, and this algorithm probably also runs in subexponential time.

Although it may not be surprising that it is possible to solve CVPP in subexponential time with (super)exponential space, it is not clear what the complexities of other methods would be. Our method presents a clear tradeoff between the complexities, where the constants in the preprocessing exponent are quite small; for instance, we can solve CVPP in time $2^{0.06d + o(d)}$ with less than $2^{d + o(d)}$ memory, which is the same amount of memory/preprocessing of the best provable SVP and CVP algorithms~\cite{aggarwal15, aggarwal15b}. Indexing the fundamental domain may well require much more memory than this. 

\subsection{Bounded Distance Decoding with Preprocessing}

We finally take a look at specific instances of CVP which are easier than the general problem, such as when the target $\vc{t}$ lies unusually close to the lattice. This problem naturally appears in practice, when a private key consists of a \textit{good basis} of a lattice with short basis vectors, and the public key is a \textit{bad basis} of the same lattice. An encryption of a message could then consist of the message being mapped to a lattice point $\vc{v} \in \cL$, and a small error vector $\vc{e}$ being added to $\vc{v}$ ($\vc{t} = \vc{v} + \vc{e}$) to hide $\vc{v}$. If the noise $\vc{e}$ is small enough, then with a good basis one can decode $\vc{t}$ to the closest lattice vector $\vc{v}$, while someone with the bad basis cannot decode correctly. As decoding for arbitrary $\vc{t}$ (solving CVP) is known to be hard even with knowledge of a good basis~\cite{micciancio01e, feige02, regev04d, alekhnovich05}, $\vc{e}$ needs to be very short, and $\vc{t}$ must lie unusually close to the lattice.

So instead of assuming target vectors $\vc{t} \in \mathbb{R}^d$ are sampled at random, suppose that $\vc{t}$ lies at distance at most $\delta \cdot \lambda_1(\cL)$ from $\cL$, for $\delta \in (0,1)$. For adaptive sieving, recall that the list size $(4/3)^{d/2 + o(d)}$ is the minimum initial list size one can hope to use to obtain a list of short lattice vectors; with fewer vectors, one would not be able to solve SVP.\footnote{The recent paper \cite{bai16} discusses how to use less memory in sieving, by using triple- or tuple-wise reductions, instead of the standard pairwise reductions. These techniques may also be applied to adaptive sieving to solve CVP with less memory, at the cost of an increase in the time complexity.} For non-adaptive sieving however, it may be possible to reduce the list size below $2^{d/2 + o(d)}$.

\subsubsection{List size.} Let us again assume that the preprocessed list $L$ contains almost all $\alpha^{d + o(d)}$ lattice vectors of norm at most $\alpha \cdot \lambda_1(\cL)$. The choice of $\alpha$ implies a maximum norm $\beta_{\alpha} \cdot \lambda_1(\cL)$ of the reduced vector $\vc{t}'$, as described in Lemma~\ref{lem:2}. The nearest lattice vector $\vc{s} \in \cL$ to $\vc{t}'$ lies within radius $\delta \cdot \lambda_1(\cL)$ of $\vc{t}'$, and w.h.p.\ $\vc{s} - \vc{t}'$ is approximately orthogonal to $\vc{t}'$; see Figure~\ref{fig:2b}, where the shaded area is asymptotically negligible. Therefore w.h.p.\ $\vc{s}$ has norm at most $(\sqrt{\beta_{\alpha}^2 + \delta^2}) \cdot \lambda_1(\cL)$. Now if $\sqrt{\beta_{\alpha}^2 + \delta^2} \leq \alpha$, then we expect the nearest vector to be contained in $L$, so that ultimately $\vc{0}$ is nearest to $\vc{t}'$. Substituting $\alpha^4 - 4 \beta^2 \alpha^2 + 4 \beta^2 = 0$ and $\beta^2 + \delta^2 \leq \alpha^2$, and solving for $\alpha$, this leads to the following condition on $\alpha$.
\begin{align}
\alpha^2 \geq \tfrac{2}{3} (1 + \delta^2) + \tfrac{2}{3} \sqrt{(1 + \delta^2)^2 - 3 \delta^2} \, . \label{eq:a2}
\end{align}
Taking $\delta = 1$, corresponding to exact CVP, leads to the condition $\alpha \geq \sqrt{2}$ as expected, while in the limiting case of $\delta \to 0$ we obtain the condition $\alpha \geq \sqrt{4/3}$. This matches experimental observations using the GaussSieve, where after finding the shortest vector, newly sampled vectors often cause \textit{collisions} (i.e.\ being reduced to the $\vc{0}$-vector). In other words, Algorithm~\ref{alg:nonadaptive} often reduces target vectors $\vc{t}$ which essentially lie \textit{on} the lattice ($\delta \to 0$) to the $\vc{0}$-vector when the list has size $(4/3)^{d/2 + o(d)}$. This explains why collisions in the GaussSieve are common when the list size grows to size $(4/3)^{d/2 + o(d)}$.

\subsubsection{Main result.} To solve BDD with a target $\vc{t}$ at distance $\delta \cdot \lambda_1(\cL)$ from the lattice, we need the preprocessing to produce a list of almost all $\alpha^{d + o(d)}$ vectors of norm at most $\alpha \cdot \lambda_1(\cL)$, with $\alpha$ satisfying~\eqref{eq:a2}. Similar to the analysis for CVP, we can produce such a list by only doing reductions between two vectors if their angle is less than $\theta$, where now $\theta = \arcsin(1 / \alpha)$. Combining this with Lemma~\ref{lem:1}, we obtain the following result.

\begin{theorem} \label{thm:BDD}
Let $\alpha$ satisfy \eqref{eq:a2} and let $u \in (\sqrt{\frac{\alpha^2 - 1}{\alpha^2}}, \sqrt{\frac{\alpha^2}{\alpha^2 - 1}})$. Using non-adaptive sieving, we can heuristically solve BDD for targets $\vc{t}$ at distance $\delta \cdot \lambda_1(\cL)$ from the lattice, with preprocessing time $\T_1$, space complexity $\Sp$, and query time complexity $\T_2$ as follows:
\begin{align}
  & \qquad \Sp = \left(\frac{1}{1 - (\alpha^2 - 1) (u^2 - \frac{2 u}{\alpha}
 \sqrt{\alpha^2 - 1} + 1)}\right)^{d/2 + o(d)}, \\
 \T_1 &= \max\left\{\Sp, \ (3/2)^{d/2 + o(d)}\right\}, \qquad  \T_2 = \left(\frac{\alpha + u \sqrt{\alpha^2 - 1}}{2 \alpha - \alpha^3 + \alpha^2 u \sqrt{\alpha^2 - 1}}\right)^{d/2 + o(d)}.
\end{align}
\end{theorem}

\begin{proof}
These complexities directly follow from applying Lemma~\ref{lem:nns} with $\theta = \arcsin(1/\alpha)$, and again observing that Phase 1 can be performed in time $\T_1 = n^{1 + \rhou}$ and space $\Sp = n^{1 + \rhou}$, while Phase 2 takes time $\T_2 = n^{\rhoq}$. Note that we cannot combine vectors whose angles are larger than $\frac{\pi}{3}$ in Phase 1, which leads to a lower bound on the preprocessing time complexity $\T_1$ based on the costs of solving SVP.
\end{proof}

Theorem~\ref{thm:BDD} is a generalization of Theorem~\ref{thm:2}, as the latter can be derived from the former by substituting $\delta = 1$ above. To illustrate the results, Figure~\ref{fig:1} considers two special cases:
\begin{itemize} 
	\item For $\delta = \frac{1}{2}$, we find $\alpha \approx 1.1976$, leading to $\Sp \approx 2^{0.2602d + o(d)}$ and $\T_2 = 2^{0.1908d + o(d)}$ when minimizing the space complexity.
	\item For $\delta \to 0$, we have $\alpha \to \sqrt{4/3} \approx 1.1547$. The minimum space complexity is therefore $\Sp = (4/3)^{d/2 + o(d)}$, with query complexity $\T_2 = 2^{0.1610d + o(d)}$.
\end{itemize}
In the limit of $u \to \sqrt{\frac{\alpha^2}{\alpha^2 - 1}}$ we need superexponential space/preprocessing $\Sp, \T_1 \to 2^{\omega(d)}$ and a subexponential query time $\T_2 \to 2^{o(d)}$ for all $\delta > 0$.
	
\subsection{Approximate Closest Vector Problem with Preprocessing}

Given a lattice $\cL$ and a target vector $\vc{t} \in \mathbb{R}^d$, approximate CVP with approximation factor $\kappa$ asks to find a vector $\vc{s} \in \cL$ such that $\|\vc{s} - \vc{t}\|$ is at most a factor $\kappa$ larger than the real distance from $\vc{t}$ to $\cL$. For random instances $\vc{t}$, by the Gaussian heuristic this means that a lattice vector counts as a solution iff it lies at distance at most $\kappa \cdot \lambda_1(\cL)$ from $\vc{t}$.

\subsubsection{List size.} Instead of reducing $\vc{t}$ to a vector $\vc{t}'$ of norm at most $\lambda_1(\cL)$ as is needed for solving exact CVP, we now want to make sure that the reduced vector $\vc{t}'$ has norm at most $\kappa \cdot \lambda_1(\cL)$. If this is the case, then the vector $\vc{t} - \vc{t}'$ is a lattice vector lying at distance at most $\kappa \cdot \lambda_1(\cL)$, which w.h.p.\ qualifies as a solution. This means that instead of substituting $\beta = 1$ in Lemma~\ref{lem:2}, we now substitute $\beta = \kappa$. This leads to the condition that $\alpha_0^4 - 4\kappa^2 \alpha_0^2 + 4 \beta^2 < 0$ for some $\alpha_0 \leq \alpha$. By a similar analysis $\alpha^2$ must therefore be larger than the smallest root $r_1 = 2\kappa (\kappa - \sqrt{\kappa^2 - 1})$ of this quadratic polynomial in $\alpha^2$. This immediately leads to the following condition on $\alpha$:
\begin{align}
\alpha^2 \geq 2 \kappa \left(\kappa - \sqrt{\kappa^2 - 1}\right). \label{eq:a3}
\end{align}
A sanity check shows that $\kappa = 1$, corresponding to exact CVP, indeed results in $\alpha \geq \sqrt{2}$, while in the limit of $\kappa \to \infty$ a value $\alpha \approx 1$ suffices to obtain a vector $\vc{t}'$ of norm at most $\kappa \cdot \lambda_1(\cL)$. In other words, to solve approximate CVP with very large (constant) approximation factors, a preprocessed list of size $(1 + \eps)^{d + o(d)}$ suffices.

\subsubsection{Main result.} Similar to the analysis of CVPP, we now take $\theta = \arcsin(1/\alpha)$ as the angle with which to reduce vectors in Phase 1, so that the output of Phase 1 is a list of almost all $\alpha^{d + o(d)}$ shortest lattice vectors of norm at most $\alpha \cdot \lambda_1(\cL)$. Using a smaller angle $\theta$ for reductions again means that nearest neighbor searching can speed up the reductions in both Phase 1 and Phase 2 even further. The exact complexities follow from Lemma~\ref{lem:nns}.

\begin{theorem} \label{thm:aCVP}
Using non-adaptive sieving with spherical LSF, we can heuristically solve $\kappa$-CVP with similar complexities as in Theorem~\ref{thm:BDD}, where now $\alpha$ must satisfy \eqref{eq:a3}.
\end{theorem}

Note that only the dependence of $\alpha$ on $\kappa$ is different, compared to the dependence of $\alpha$ on $\delta$ for bounded distance decoding. The complexities for $\kappa$-CVP arguably decrease \textit{faster} than for $\delta$-BDD: for instance, for $\kappa \approx 1.0882$ we obtain the same complexities as for BDD with $\delta = \frac{1}{2}$, while $\kappa = \sqrt{4/3} \approx 1.1547$ leads to the same complexities as for BDD with $\delta \to 0$. Two further examples are illustrated in Figure~\ref{fig:1}:

\begin{itemize}
	\item For $\kappa = 2$, we have $\alpha \approx 1.1976$, which for $u \approx 0.5503$ leads to $\Sp = \T_1 = 2^{0.2602 d + o(d)}$ and $\T_2 = 2^{0.1908 d + o(d)}$, and for $u = 1$ leads to $\Sp = \T_1 = 2^{0.3573 d + o(d)}$ and $\T_2 = 2^{0.0971 d + o(d)}$.
	\item For $\kappa \to \infty$, we have $\alpha \to 1$, i.e.\ the required preprocessed list size approaches $2^{o(d)}$ as $\kappa$ grows. For sufficiently large $\kappa$, we can solve $\kappa$-CVP with a preprocessed list of size $2^{\eps d + o(d)}$ in at most $2^{\eps d + o(d)}$ time. The preprocessing time is given by $2^{0.2925 d + o(d)}$.
\end{itemize}

The latter result shows that for any superconstant approximation factor $\kappa = \omega(1)$, we can solve the corresponding approximate closest vector problem with preprocessing in subexponential time, with an exponential preprocessing time complexity $2^{0.292d + o(d)}$ for solving SVP and generating a list of short lattice vectors, and a subexponential space complexity required for Phase 2. In other words, even without superexponential preprocessing/memory we can solve CVPP with large approximation factors in subexponential time.

To compare this result with previous work, note that the lower bound on $\alpha$ from \eqref{eq:a3} tends to $1 + 1/(8 \kappa^2) + O(\kappa^{-4})$ as $\kappa$ grows. The query space and time complexities are further both proportional to $\alpha^{\Theta(d)}$. To obtain a polynomial query complexity and polynomial storage after the preprocessing phase, we can solve for $\kappa$, leading to the following result.

\begin{corollary} \label{cor:acvpp-poly}
With non-adaptive sieving we can heuristically solve approximate CVPP with approximation factor $\kappa$ in polynomial time with polynomial-sized advice iff $\kappa = \Omega(\sqrt{d / \log d})$.
\end{corollary}

\begin{proof}
The query time and space complexities are given by $\alpha^{\Theta(d)}$, where $\alpha = 1 + \Theta(1 / \kappa^2)$. To obtain polynomial complexities in $d$, we must have $\alpha^{\Theta(d)} = d^{O(1)}$, or equivalently:
\begin{align}
1 + \Theta\left(\frac{1}{\kappa^2}\right) = \alpha = d^{O(1/d)} = \exp \, O\left(\frac{\log d}{d}\right) = 1 + O\left(\frac{\log d}{d}\right).
\end{align}
Solving for $\kappa$ leads to the given relation between $\kappa$ and $d$.
\end{proof}

Apart from the heuristic assumptions, this approximation factor is equivalent to Aharonov and Regev~\cite{aharonov04}, who showed that the decision version of CVPP with approximation factor $\kappa = \Omega(\sqrt{d / \log d})$ can provably be solved in polynomial time. This further (heuristically) improves upon results of~\cite{lagarias90b, dadush14}, who are able to solve search-CVPP with polynomial time and space complexities for $\kappa = O(d^{3/2})$ and $\kappa = \Omega(d / \sqrt{\log d})$ respectively. Assuming the heuristic assumptions are valid, Corollary~\ref{cor:acvpp-poly} closes the gap between these previous results for decision-CVPP and search-CVPP with a rather simple algorithm: (1) preprocess the lattice by storing all $d^{O(1)}$ shortest vectors of the lattice in a list; and (2) apply Algorithm~\ref{alg:nonadaptive} to this list and the target vector to find an approximate closest vector. Note that nearest neighbor techniques only affect leading constants; even without nearest neighbor searching this would heuristically result in a polynomial time and space algorithm for $\kappa$-CVPP with $\kappa = \Omega(\sqrt{d / \log d})$. An interesting open problem would be to see if this result can be made provable for arbitrary lattices, without any heuristic assumptions.


\section*{Acknowledgments}

The author is indebted to L\'{e}o Ducas, whose initial ideas and suggestions on this topic motivated work on this paper. The author further thanks Vadim Lyubashevsky and Oded Regev for their comments on the relevance of a subexponential time CVPP algorithm requiring (super)exponential space. The author is supported by the SNSF ERC Transfer Grant CRETP2-166734 FELICITY.


\bibliographystyle{alpha}
\bibliography{SACDatabase}


\appendix

\section{Pseudocode of SVP algorithms}
\label{app:alg}

Algorithms~\ref{alg:nv} and \ref{alg:gauss} present pseudo-code for the (sieve part of the) original Nguyen-Vidick sieve and the GaussSieve, respectively, as described in Section~\ref{sec:pre}. For the Nguyen-Vidick sieve, the presented algorithm is a more intuitive but equivalent version of the original sieve; see~\cite[Appendix B]{laarhoven15crypto} for details on this equivalence.

\begin{algorithm}[!h]
\caption{The quadratic Nguyen-Vidick sieve for finding shortest vectors}
\label{alg:nv}
\begin{algorithmic}[1]
\Require An input list $L \subset \cL$ of $(4/3)^{d/2 + o(d)}$ vectors of norm at most $R$
\Ensure The output list $L' \subset \cL$ has $(4/3)^{d/2 + o(d)}$ vectors of norm at most $\gamma \cdot R$
\State Initialize an empty list $L'$ 
\For{\textbf{each} $\vc{w}_1, \vc{w}_2 \in L$}
	\If{$\|\vc{w}_1 - \vc{w}_2\| \leq \gamma R$}
		\State Add $\vc{w}_1 - \vc{w}_2$ to the list $L'$
	\EndIf
\EndFor
\State \Return $L'$
\end{algorithmic}
\end{algorithm}

\begin{algorithm}[!h]
\caption{The GaussSieve algorithm for finding shortest vectors}
\label{alg:gauss}
\begin{algorithmic}[1]
\Require A basis $B$ of a lattice $\cL(B)$ 
\Ensure The algorithm returns a shortest lattice vector
\State Initialize an empty list $L$ and an empty stack $S$
\Repeat
	\State Get a vector $\vc{v}$ from the stack (or sample a new one if $S = \emptyset$) 
	\For{\textbf{each} $\vc{w} \in L$}
		\If{$\|\vc{v} - \vc{w}\| \leq \|\vc{v}\|$}
			\State Replace $\vc{v} \leftarrow \vc{v} - \vc{w}$
		\EndIf
		\If{$\|\vc{w} - \vc{v}\| \leq \|\vc{w}\|$}
			\State Replace $\vc{w} \leftarrow \vc{w} - \vc{v}$
			\State Move $\vc{w}$ from the list $L$ to the stack $S$ (unless $\vc{w} = \vc{0}$)
		\EndIf
	\EndFor
	\If{$\vc{v}$ has changed}
		\State Add $\vc{v}$ to the stack $S$ (unless $\vc{v} = \vc{0}$)
	\Else
		\State Add $\vc{v}$ to the list $L$ (unless $\vc{v} = \vc{0}$)
	\EndIf
\Until{$\vc{v}$ is a shortest vector}
\State \Return $\vc{v}$
\end{algorithmic}
\end{algorithm}

\section{Pseudocode of Phase 1 for non-adaptive sieving}
\label{app:alg2}

To generate a list of the $\alpha^{d + o(d)}$ shortest lattice vectors with the GaussSieve, rather than the $(4/3)^{d/2 + o(d)}$ lattice vectors one would get with standard sieving, we relax the reductions: reducing if $\|\vc{v} - \vc{w}\| < \|\vc{v}\|$ corresponds to an angle $\pi/3$ between $\vc{v}$ and $\vc{w}$, leading to a list size $(1/\sin(\frac{\pi}{3}))^{d + o(d)} = (4/3)^{d/2 + o(d)}$. To obtain a list of size $\alpha^{d + o(d)}$, we reduce vectors if their angle is less than $\theta = \arcsin(1/\alpha)$, which for vectors $\vc{v}, \vc{w}$ of similar norm corresponds to the following condition:
\begin{align}
\|\vc{v} - \vc{w}\| < \sqrt{2 (1 - \cos \theta)} \cdot \|\vc{v}\| = \sqrt{2 - \frac{2}{\alpha} \sqrt{\alpha^2 - 1}} \cdot \|\vc{v}\|.
\end{align}
This leads to the modified GaussSieve described in Algorithm~\ref{alg:nonadaptive0}.

\begin{algorithm}[!t]
\caption{The non-adaptive GaussSieve (Phase 1) for finding closest vectors}
\label{alg:nonadaptive0}
\begin{algorithmic}[1]
\Require A basis $B$ of a lattice $\cL(B)$, a parameter $\alpha > 1$ 
\Ensure The output list $L$ contains $\alpha^{d + o(d)}$ vectors of norm at most $\alpha \cdot \lambda_1(\cL)$
\State Initialize an empty list $L$ and an empty stack $S$
\State Let $\alpha_0 = \max\{\alpha, \sqrt{4/3}\}$
\Repeat
	\State Get a vector $\vc{v}$ from the stack (or sample a new one if $S = \emptyset$) 
	\For{\textbf{each} $\vc{w} \in L$}
		\If{$\|\vc{v} - \vc{w}\|^2 \leq (2 - \frac{2}{\alpha_0} \sqrt{\alpha_0^2 - 1}) \cdot \|\vc{v}\|^2$}
			\State Replace $\vc{v} \leftarrow \vc{v} - \vc{w}$
		\EndIf
		\If{$\|\vc{w} - \vc{v}\|^2 \leq (2 - \frac{2}{\alpha_0} \sqrt{\alpha_0^2 - 1}) \cdot \|\vc{w}\|^2$}
			\State Replace $\vc{w} \leftarrow \vc{w} - \vc{v}$
			\State Move $\vc{w}$ from the list $L$ to the stack $S$ (unless $\vc{w} = \vc{0}$)
		\EndIf
	\EndFor
	\If{$\vc{v}$ has changed}
		\State Add $\vc{v}$ to the stack $S$ (unless $\vc{v} = \vc{0}$)
	\Else
		\State Add $\vc{v}$ to the list $L$ (unless $\vc{v} = \vc{0}$)
	\EndIf
\Until{$\vc{v}$ is a shortest vector}
\State \Return $L$
\end{algorithmic}
\end{algorithm}

\end{document}